\documentclass[preprint,11pt]{elsarticle}

\usepackage{amssymb}
 \usepackage{amsthm}

\usepackage{amsfonts}
\usepackage{amsmath}
\usepackage{amssymb}
\usepackage{hyperref}

\newtheorem{theorem}{Theorem}
\newtheorem{lemma}{Lemma}
\newtheorem{definition}{Definition}
\newtheorem{corollary}{Corollary}
\newtheorem{proposition}{Proposition}

\begin{document}

\begin{frontmatter}



\title{Relating counting complexity to non-uniform probability measures}


\author{Eleni Bakali}

\address{National Technical University of Athens, \\
School of Electrical and Computer Engineering, 
Department of Computer Science, \\
Computation and Reasoning Laboratory (CoReLab, room 1.1.3),\\ 
9 Iroon Polytechniou St, Polytechnioupoli Zografou,  157 80,  Athens, Greece.}

\ead{mpakalitsa[at]gmail.com}

\begin{abstract}
 A standard method for designing randomized algorithms to approximately count the number of solutions of a problem in \#P, is by constructing a rapidly mixing Markov chain converging to the uniform distribution over this set of solutions. This construction is not always an easy task, and it is conjectured that it is not always possible. We want to investigate other possibilities for using Markov Chains in relation to counting, and whether we can relate algorithmic counting to other, non-uniform, probability distributions over the set we want to count. 

In this paper we present a family of probability distributions over the set of solutions of a problem in TotP, and show how they relate to counting; counting is equivalent to computing their normalizing factors. 
We analyse the complexity of sampling, of computing the normalizing factor, and of computing the size support of these distributions. The latter is also equivalent to counting. We also show how the above tasks relate to each other, and to other problems in complexity theory as well. 

In short,  we prove that sampling and approximating the normalizing factor is easy. We do this by constructing a family of rapidly mixing Markov chains for which these distributions are stationary.  At the same time we show that exactly computing the normalizing factor is TotP-hard. However the reduction proving the latter is not approximation preserving, which conforms with the fact that TotP-hard problems are inapproximable if NP $\neq$ RP.

The problem we consider is the Size-of-Subtree, a TotP-complete problem under parsimonious reductions. Therefore the results presented here extend to any problem in TotP. TotP is the Karp-closure of self-reducible problems in \#P, having decision version in P.
\end{abstract}


\end{frontmatter}

\section{Introduction}

The set of all self-reducible counting problems in \#P, having decision version in P, is contained in a complexity class called TotP \cite{PZ06}. E.g. \#DNF, \#IS, and the Permanent belong to TotP. 
TotP is a proper subclass of \#P unless P=NP \cite{PZ06}.
Regarding approximability, TotP admits  FPRAS if and only if  NP=RP (\cite{DFJ02} for the one direction, and theorem \ref{theorem2} for the other).

We are (I am) in particularly interested  in understanding whether NP=RP, through the lens of counting complexity. In other words we are interested in understanding whether there exist counting problems in TotP, that are hard to approximate, and if so, what the reason for this difficulty is.

More specifically in this paper, the motivation derives from a theorem \cite{JS89} asserting that approximate counting is \textit{equivalent} to\textit{ uniform} sampling.  Therefore counting is most often performed via uniform sampling. Sampling in turn is usually performed by designing a rapidly mixing Markov chain, having as stationary distribution the uniform  over the set of solutions to the problem at hand. This is a special case of the Markov Chain Monte Carlo method (MCMC), applied to counting problems.

In order to study the possibility of approximating the problems in TotP, one can focus on any TotP-complete problem, under parsimonious reductions (i.e. reductions that preserve the value of the function, and thus also preserve approximations). Recently we found some such complete problems \cite{BCPPZ17}, and one of them is the Size-Of-Subtree problem, also known as the problem of estimating the size of a backtracking-tree \cite{ K74, Stockmeyer85a}. This problem asks for an estimation of the size of some sub-tree of the perfect binary tree, with height $n$, given in some succinct (polynomial in $n$) representation. 

We first observe that uniform sampling over the nodes of a tree can be performed by a simple random walk on the tree, but its mixing time is quadratic to the size of the tree, and thus in worst case exponential in $n$. This holds unconditionally, i.e. independently of whether NP = RP. However it might be possible to perform uniform sampling differently (open question). 

A different direction, which we do follow in this paper, is to explore whether it is possible to associate counting in TotP to other probability measures over the set of solutions we want to count, (we mean other than the uniform). In particular, we want to investigate whether counting can be related to computational tasks that concern these probability measures, as sampling, estimating the normalizing factor, and estimating the support.  For example, it was natural (for me) to wonder whether there exist a rapidly mixing Markov chain over the nodes of a tree, such that its stationary distribution can be somehow related to exactly or approximately computing the size of the tree. 

In this paper, in short, it is shown that there indeed exists a family of probability distributions related to counting in TotP, in the sense that counting in TotP is equivalent to computing the normalizing factors of these distributions.
We analyse the complexity of the above mentioned three tasks. We also show how these three tasks relate to each other and to other problems in complexity theory.
We discuss the results in detail later.

Before proceeding to the presentation of the  results and proofs, we would like to provide an indication that it might be easier and more fruitful to study problems in TotP, rather than in \#P, and in particular to try to use Markov chains for them, and study their stationary distributions.

Take for example \#SAT. 
From the study of random SAT it is known that, for  formulas considered hard, the satisfying assignments are widely scattered in the space of all assignments \cite{A09,AR09,ACR10}. 
That is,  the solutions form clusters  s.t. it is not only hard to find even one solution, but even if you are given one solution in some cluster, it is hard to find another one in a different cluster. The solutions doesn't seem to relate to each other in any algorithmically tractable way.
On the contrary, for the TotP-complete version of \#SAT (see \cite{BCPPZ17} for details), the situation is completely different. The solutions of an input formula build an easy neighbourhood structure, in particular a tree of polynomial height, s.t. from any solution, it is easy to move to its neighbouring solutions. 
This property of connectedness seems to be a main difference between problems in TotP and \#P. In fact it can be generalized to any problem in TotP. 

From an algorithmic point of view, connectedness is important, because it allows us to easily construct irreducible Markov chains. Moreover, from the hardness point of view,
the observation of scattered solutions can be a reasonable explanation for the failure of many algorithmic approaches for \#SAT, but not for its TotP-complete version. 

Of course \#P admits an FPRAS if and only if TotP admits an FPRAS (as we can see from the discussion in the beginning of this section), however at first glance  it might be easier to design an approximation algorithm for a TotP-complete problem, than it is for a \#P-complete problem, if such an algorithm exists. If on the other hand NP $\neq$ RP, and thus FPRAS is impossible for both \#P and TotP, new insights and explanations are needed, that apply to TotP, not only to \#P. 

\section{Main Results}
In this paper a family of probability distributions $(\pi_S)_{_S}$ defined on the set of nodes of finite binary trees is presented, and the complexity of sampling, computing the normalizing factor, and computing the size of the support is studied. Time complexity is considered with respect to the height $n$ of the corresponding tree. 

The family of distributions $(\pi_S)_{_S}$ is defined as follows. 

\begin{definition}Let $S$ be a binary tree of height $n$ and let $V(S)$ be the set of nodes of $S.$ For all $u \in V(S)$ $\pi_S(u)=\alpha \cdot 2^{n-i},$ where $i$ is the depth of node $u$ and $\alpha$ is the normalizing factor of $\pi_S$, that is a constant s.t. $\sum_{u\in V(S)} \pi_S(u)=1.$  
\label{def-PS}\end{definition}

The following are shown. 

\begin{theorem}
\begin{enumerate}
\item For every distribution in this family, there is a Markov chain with polynomial mixing time, having this distribution as stationary.
\item For every distribution in this family, sampling with respect to this distribution is possible in randomized polynomial time, using the above Markov Chain.
\item Computing the normalizing factor of any distribution in this family 
	\begin{enumerate}
	\item is TotP-hard under Turing reductions,
	\item approximation is possible with FPRAS, using sampling,
	\item exact computation is impossible deterministically (or respectively probabilistically) if NP $\neq$ P (or respectively NP $\neq$ RP).
	\end{enumerate}
\item Computing the size of the support of any distribution in this family
	\begin{enumerate}
	\item \label{4a} is TotP-hard under parsimonious reductions,
	\item \label{4b} reduces to exactly computing the normalizing factors,
	\item additive error approximation (see def. \ref{additive definition}) is possible in randomized polynomial time, (note that the size of the support is in general exponential in $n$),
	\item exact computation is impossible deterministically (or respectively probabilistically) if NP $\neq$ P (or respectively NP $\neq$ RP),
	\item multiplicative polynomial factor approximation is impossible deterministically (or respectively probabilistically) iff NP $\neq$ P (or respectively NP $\neq$ RP),
	\item \label{4last} additive error approximation can also be achieved in randomized polynomial time, by reducing it to approximately computing normalizing factors for any distribution in the family.
	\end{enumerate}
\end{enumerate} 
\label{main-theorem}\end{theorem}

\subsection{Main ideas and techniques}

A simple random walk on a perfect binary tree mixes in exponential time w.r.t the height of the tree. An intuitive reason for this, is that from any internal node, the probability of going a level down is double the probability of going a level up (since it has got two children, but one parent). The idea is to construct a Markov chain s.t. the probability of going a level up equals the probability of going a level down. This Markov chain turns out to converge rapidly on the full complete binary tree, and as we prove, this also generalises to an arbitrary binary tree $S$ (not full or complete).

Of course this Markov chain does not converge to the uniform distribution over the set of nodes $V(S).$ However the normalizing factor of its stationary distribution contains enough information to compute $|V(S)|.$  More precisely, if you consider the pruned subtrees $S_0, S_1, ..., S_n,$ where each $S_i$ contains all nodes of $S$ up to depth $i,$  then the corresponding normalizing factors $\alpha_{S_0}, \alpha_{S_1}, ... \alpha_{S_i}$ can be used to iteratively count the number of nodes at level $i.$

Our techniques refer to the analysis of Markov chains. We prove polynomial mixing time, by bounding a quantity called conductance.

\subsection{Proofs}
The counting problem considered here is the Size-of-Subtree.
\begin{definition} Let $T$ be the perfect binary tree of height $n.$ Let $S$ be a subtree of height $n$ containing the root of $T$,  given in succinct representation, e.g. by a polynomial computable predicate $R_S:V(T)\rightarrow\{0,1\}$ s.t. $R_S(u)=1$ iff $u\in V(S),$ where $V(\cdot)$ denotes the set of vertices. The counting problem Size-of-Subtree is to compute the size of $V(S)$.
\end{definition}

This problem is TotP-complete under parsimonious reductions \cite{BCPPZ17}, so the results extend to any problem in TotP. For an arbitrary problem \#A in TotP, there is a tree $S$ s.t. its nodes are in one-to-one correspondence with the solutions of A. We will not get into that. For more details see \cite{BCPPZ17}.

In order to prove theorem \ref{main-theorem}, we will first prove some propositions.

We define a family of Markov chains $(P_S)_{_S},$ each having as states the nodes of a binary tree $S$.

\begin{definition}\label{the_chain} Let $S$ be a subtree of the perfect binary tree $T$ of height $n$, containing the root of $T$. We define the Markov chain $P_S$ over the nodes of $S$, with the following  transition probabilities. \\$p_S(i,j)=1/2$ if $j$ is the parent of $i$, \\$p_S(i,j)=1/4$ if $j$ is a child of $i$, \\$p_S(i,j)=0$ for every other $j\neq i$, and \\$p_S(i,i)=1-\sum_{j\neq i}p_S(i,j)$.
\end{definition}

\begin{proposition}\label{stationary}
The stationary distribution of $P_S$ is $\pi_S$ as defined in \ref{def-PS}.
\end{proposition}

\begin{proof}
It is easy to check that $\sum_{i}\pi_S(i)p_S(i,j)=\pi_S(j).$ 
\end{proof}

\textbf{Note} For simplicity of notations from now on we will assume that S is fixed and  omit it from the subscripts in $P_S, p_S, \pi_S$.

Now we will prove that $P$ is rapidly mixing, i.e. mixes in time polynomial in the height of the tree $S$.  We will use the following  lemma from \cite{JS89}. 

Let $\{X_t\}_{t\geq 0}$ be a Markov chain over a finite state space $\cal{X}$ with transition probabilities $p_{ij}$. Let $p_x^{(t)}$ be the distribution of $X_t$ when starting from state $x$. Let $\pi$ be the stationary distribution,  and let $\tau_x(\epsilon)=\min\{t:||p_x^{(t)}-\pi ||\leq\epsilon\}$ be the mixing time when starting from state $x$. An ergodic Markov chain is called time reversible if $\forall i,j\in{\cal X}, p_{ij}\pi_i=p_{ji}\pi_j $. Let $H$ be the underlying graph of the chain, for which we have an edge with weight $w_{ij}=p_{ij}\pi_i=p_{ji}\pi_j$ for each $i,j\in\cal{X}$. A Markov chain is called lazy if $\forall i\in{\cal X},p_{ii}\geq \frac{1}{2} .$ In \cite{JS89} the conductance of a time reversible Markov chain is defined, as follows: $\Phi(H)=\min\frac{\sum_{i\in Y,j\notin Y}w_{ij}}{\sum_{i \in Y}\pi_{i}}$, where the minimum is taken over all $Y\subseteq {\cal X}$ s.t. $0<\sum_{i\in Y}\pi_i\leq\frac{1}{2}.$

\begin{lemma}
\cite{JS89} For any lazy, time reversible Markov chain \[\tau_{x}(\epsilon)\leq  const \times \left[ \frac{1}{\Phi(H)^2}(\log\pi_x^{-1}+\log\epsilon^{-1}) \right].\]
\end{lemma}

\begin{proposition} \label{mixing time}
The mixing time of $P$, when starting from the root, is polynomial in the height of the tree $n$.
\end{proposition}  
\begin{proof} First of all, we will consider the lazy version of the Markov chain, i.e. in every step, with probability $1/2$ we do nothing, and with probability $1/2$ we follow the rules as in definition \ref{the_chain}. The mixing time of $P$ is bounded by the mixing time of its lazy version. The stationary distribution is the same. The Markov chain is time reversible, and the underlying graph is a tree with edge weights $w_{uv}=\pi_u p_{uv}=2^i\alpha\times \frac{1}{8}=2^{i-3}\alpha,$ if we suppose that $u$ is the father of $v$ and $2^i\alpha$ is the probability $\pi_u$.

The quantity $\pi_{root}^{-1}$ is $O(n)$, as we will show in lemma \ref{propertiesOfP}. 

Now it suffices to show that $1/\Phi(H)$ is polynomial in $n$. 

Let ${\cal X}$ be the set of the nodes of $S$, i.e. the state space of the Markov chain $P$. We will consider all possible $Y\subseteq {\cal X}$ with $0\leq \pi(Y)\leq 1/2.$ We will bound the quantity $\frac{\sum_{i\in Y,j \notin Y}w_{ij}}{\sum_{i \in Y}\pi_{i}}.$

If $Y$ is connected and does not contain the root of $S$, then it is a subtree of $S$, with root let say u, and $\pi_u=\alpha 2^k$ for some $k\in \mathbb{N}.$  We have 
\[\sum_{i\in Y, j\notin Y} w_{ij}\geq w_{u,father(u)}=2^{k-2}\alpha.\]
Now let $Y'$ be the perfect binary tree with root $u$ and height the same as $Y$, i.e. $k$. We have
\[\sum_{i\in Y} \pi_i \leq \sum_{i\in Y'}\pi_i= \sum_{j=0}^{k}2^{k-j}\alpha\times 2^j=\]\[2^k(k+1)\alpha\leq 2^k(n+1)\alpha \] where this comes if we sum over the levels of the tree $Y'.$
So it holds
\[\frac{\sum w_{ij}}{\sum \pi_i}\geq \frac{2^{k-2}\alpha}{2^k(n+1)\alpha}=\frac{1}{4(n+1)}\]

If $Y$ is the union of two subtrees of $S$, not containing the root of $S$, and the root of the first is an ancestor of the second's root, then the same arguments hold, where now we take as $u$ the root of the first subtree.

If $Y$ is the union of $\lambda$ subtrees not containing the root of $S$, for which it holds that no one's root is an ancestor of any other's root, then we can prove a same bound as follows. Let $Y_1,...Y_{\lambda}$ be the subtrees, and let $k_1,k_2,...,k_{\lambda},$ be the respective exponents in the probabilities of the roots of them, in the stationary distribution. Then as before
\[\sum w_{ij}\geq 2^{k_1-2}\alpha+2^{k_2-2}\alpha+...+2^{k_{\lambda}-2}\alpha\] and
\[\sum_{i\in Y} \pi_{i}=\sum_{j=1...\lambda}\sum_{i\in Y_j} \pi_i\leq 2^{k_j}(n+1)\alpha\] thus
\[\frac{\sum w_{ij}}{\sum \pi_i}\geq \frac{\alpha\sum_{j=1...\lambda} 2^{k_j-2}}{(n+1)\alpha \sum_{j=1...\lambda} 2^{k_j}}=\frac{1}{4(n+1)}.\]

If $Y$ is a subtree of $S$ containing the root of $S$, then the complement of $Y$, i.e. $S\setminus Y$ is the union of $\lambda$ subtrees of the previous form. So if we let $Y_i,k_i$ be as before, then
\[\sum w_{ij}=\alpha\sum_{j=1...\lambda} 2^{k_j-2}\] and since from hypothesis $\pi(Y)\leq 1/2$, we have
\[\sum_{i\in Y}\pi_i\leq\sum_{i\in S\setminus Y} \pi_i\leq (n+1)\alpha\sum_{j=1...\lambda} 2^{k_j}\]
thus the same bound holds again.

Finally, similar arguments imply the same bound when $Y$ is an arbitrary subset of $S$ i.e. an arbitrary union of subtrees of $S$.

In total we have
$1/\Phi(H)\leq 4(n+1).$
\end{proof}

Note that this result implies mixing time $O(n^2 \log n).$ This agrees with the intuition that on the full binary tree, the mixing time should be as much as the mixing time of a simple random walk over the levels of the tree, i.e. over a chain of length $n$. The bound is looser only by a $\log n$ factor.
 
The following  lemma proves two properties of this Markov chain, needed for the proofs that will follow.

\begin{lemma}
Let $R$ be a binary tree of height $n$, and let $\alpha_R$ be the normalizing factor of the stationary distribution $\pi_{R}$ of the above Markov chain. It holds $\alpha_R^{-1}\leq (n+1)2^n,$ and $\pi_R(root)\geq \frac{1}{n+1}$
\label{propertiesOfP}
\end{lemma}
\begin{proof}
Let $r_i$ be the number of nodes in depth $i$.
\[1=\sum_{u\in S}\pi_R(u)=\sum_{i=0}^n\sum_{u\ in\ level\ i}\pi_R(u)=\sum_{i=0}^n r_i\alpha_R\cdot 2^{n-i}\] \[
\Rightarrow \frac{1}{\alpha_R}=\sum_{i=0}^n r_i\cdot 2^{n-i}\] which is maximized when the $r_i$'s are maximized, i.e. when the tree is perfect binary, in which case $r_i=2^i$ and $\alpha_R^{-1}=(n+1)2^n.$ This also implies that for the root of  $R$ it holds $\pi_R(root)=\alpha_R\cdot 2^n\geq \frac{1}{n+1}.$ 
\end{proof}

Now we will reduce the computation of the size of $S$ to the computation of the normalizing factors of the above probability distributions $(\pi_S)_{_S}.$

\begin{proposition} \label{sizeOfS}
Let $S$ be a binary tree of height $n$, and $\forall i=0...n,$ let $S_i$ be the subtree of  $S$ that contains all nodes up to depth $i$, and let $\alpha_{S_i}$ be the corresponding normalizing factors defined as above. Then  
\[|S|=\frac{1}{\alpha_{S_n}}-\sum_{k=0}^{n-1}\frac{1}{\alpha_{S_k}}\]
\end{proposition}

\begin{proof}
For $i=1,...,n$ let $r_i$ be the number of nodes in depth $i$. So 
$|S|=r_0+...+r_n.$

Obviously if $S$ is not empty,
\begin{equation}
r_0=1=\frac{1}{\alpha_{S_0}}.\label{r0}
\end{equation}

We will prove that $\forall k=1...n$
\begin{equation}\label{rk} r_k=\frac{1}{\alpha_{S_k}}-2\frac{1}{\alpha_{S_{k-1}}},
\end{equation}
 so then $|S|=\frac{1}{\alpha_{S_0}}+\sum_{k=1}^n(\frac{1}{\alpha_{S_k}}-2\frac{1}{\alpha_{S_{k-1}}})=$ \\$\frac{1}{\alpha_{S_n}}-\sum_{k=0}^{n-1}\frac{1}{\alpha_{S_k}}.$
 
We will prove claim (\ref{rk}) by induction.

For $k=1$ we have 
\[\sum_{u\in S_{1}}\pi_{S_{1}}(u)=1\Rightarrow
\alpha_{S_1}\cdot r_1+2 \alpha_{S_1}\cdot r_0=1\Rightarrow\]
\[r_1=\frac{1}{\alpha_{S_1}}-2 r_0=\frac{1}{\alpha_{S_1}}-2\frac{1}{\alpha_{S_0}}.\]

Suppose claim (\ref{rk}) holds for $k<i\leq n.$ We will prove it holds for $k=i.$

\[\sum_{u\in S_{i}}\pi_{S_{i}}(u)=1\Rightarrow
\sum_{k=0}^i 2^{i-k} \alpha_{S_i}\cdot r_k =1 \Rightarrow\]
\[r_i=\frac{1}{\alpha_{S_i}}-\sum_{k=0}^{i-1}2^{i-k} r_k\]
and substituting $r_k$ for $k=0,...,i-1$ by (\ref{r0}) and (\ref{rk}), we get
$r_i=\frac{1}{\alpha_{S_i}}-2\frac{1}{\alpha_{S_{i-1}}}.$ 
\end{proof}

Now we give an FPRAS for the computation of the normalizing factor $\alpha_S,$ using the previously defined Markov chain $P_S$ .

\begin{proposition}\label{estimOfaR}
For any binary tree $R$ of height $n$ we can estimate $\alpha_R$, within $(1\pm\zeta)$ for any $\zeta>0$, with probability $1-\delta$ for any $\delta>0$, in time $poly(n,\zeta^{-1},\log\delta^{-1})$.
\end{proposition}

\begin{proof}
Let $R$ be a binary tree of height $n$. We can estimate $\alpha_R$ as follows.

As we saw, $\pi_R(root)=2^n \alpha_R$, and we observe that this is always at least $\frac{1}{n+1}$ (which is the case when $R$ is full binary). So we can estimate $\pi_R(root)$ within  $(1\pm\zeta)$ for any $\zeta>0$, by sampling $m$ nodes of $R$ according to $\pi_R$ and taking, as estimate, the fraction $\hat{p}=\sum_{i=1}^m\frac{1}{m}X_{i}$, where $X_i=1$ if the $i$-th sample node was the root, else $X_i=0.$

It is known by standard variance analysis arguments that we need $m=O(\pi_R(root)\cdot \zeta^{-2})=poly(n,\zeta^{-1})$ to get \[\Pr [(1-\zeta)\pi_R(root)\leq \hat{p}\leq (1+\zeta)\pi_R(root)]\geq\frac{3}{4}\]

We can boost up this probability to $1-\delta$ for any $\delta>0$, by repeating the above sampling procedure $t=O(\log\delta^{-1})$ times, and taking as final estimate the median of the $t$ estimates computed each time.

(Proofs for the above arguments are elementary in courses on probabilistic algorithms or statistics, see e.g. in \cite{Snotes} the unbiased estimator theorem and the median trick, for detailed proofs.)

The random sampling  according to $\pi_R$ can be performed by running the Markov chain defined earlier, on the nodes of $R$. Observe that the deviation $\epsilon$ from the stationary distribution can be negligible and be absorbed into $\zeta$, with only a polynomial increase in the running time of the Markov chain.

Finally, the estimate for $\alpha_R$ is $\hat{\alpha_R}=2^{-n}\hat{p}$, and it holds
\[\Pr[(1-\zeta)\alpha_R\leq \hat{\alpha_R}\leq (1+\zeta)\alpha_R]\geq 1-\delta.\]   
\end{proof}

Finally we show how the above propositions yield a probabilistic additive approximation to the problem Size-of-Subtree, although it could be also obtained by a simple random sampling process that chooses $m=poly(n)$ nodes of the full binary tree $T$ of height $n$ uniformly at random, and taking as estimate of the size of $S,$ the proportion of those $m$ samples that belong to $S.$ This is an application of the general method of \cite{Goldreich08} chapter 6.2.2.  The significance of our alternative method is related to the CAPE problem (def \ref{additive definition}), and we discuss it in section \ref{discussion}. 

\begin{definition}\label{additive definition}
We call \textbf{additive approximation} to a probability $p$, a number $\hat{p}=p \pm \xi,$ for some $\xi\in (0,1).$ In the case of Size-of-Subtree  the quantity under consideration is $p\equiv|S|/2^n.$ In the case of the Circuit Acceptance Probability Estimation problem (CAPE)\cite{Will13} the quantity under consideration is $p\equiv \#\text{sat. assignments} /2^n,$ where $n$ is the number of input gates of the given circuit. 
\end{definition}

\begin{proposition}\label{main}
For all $\xi>0,\delta>0$ we can get an estimate $|\hat{S}|$ of $|S|$ in time $poly(n,\xi^{-1},\log\delta^{-1})$ s.t. \[\Pr[|S|-\xi 2^n\leq |\hat{S}|\leq |S|+\xi 2^n]\geq1-\delta\]
\end{proposition}
\begin{proof}
Let $\zeta=\frac{\xi}{2(n+1)}$ and $\epsilon=\frac{\zeta}{1+\zeta}$, thus 
$poly(\epsilon^{-1})=poly(\zeta^{-1})$
$=poly(n,\xi^{-1})$. 

So according to proposition \ref{estimOfaR} we have in time $poly(n,\xi^{-1},\log\delta^{-1})$ estimations $\forall i=1,...,n$
\begin{equation}\label{est-aSi} (1-\epsilon)\alpha_{S_i} \leq\hat{\alpha}_{S_i} \leq (1+\epsilon) \alpha_{S_i}.
\end{equation}

We will use proposition \ref{sizeOfS}. Let $A=\frac{1}{\alpha_{S_n}}$ and $B=\sum_{k=0}^{n-1}\frac{1}{\alpha_{S_k}},$ so $|S|=A-B$, and clearly $B\leq A.$

From (\ref{est-aSi}) we have 
$\frac{1}{1+\epsilon}A\leq \hat{A} \leq \frac{1}{1-\epsilon} \Leftrightarrow$
$(1-\zeta)A\leq \hat{A} \leq (1+\zeta) A$ and similarly
$(1-\zeta)B\leq \hat{B} \leq (1+\zeta) B.$ 

Thus 
$(1-\zeta)A-(1+\zeta)B\leq \hat{A}-\hat{B} \leq (1+\zeta)A-(1-\zeta)B \Leftrightarrow$

$A-B-\zeta (A+B) \leq \hat{A}-\hat{B} \leq A-B+\zeta (A+B),$ and since $A\geq B$, we have 

$|S|-2\zeta A\leq |\hat{S}| \leq |S|+2 \zeta A.$ And since from lemma \ref{propertiesOfP} the maximum $A$ is $2^n(n+1)$, we have 

$|S|-2\zeta(n+1)2^n \leq |\hat{S}| \leq|S|+ 2\zeta (n+1) 2^n \Leftrightarrow$

$|S|-\xi \cdot 2^n \leq |\hat{S}|\leq |S|+ \xi \cdot 2^n.$ 

\end{proof}

\begin{corollary}\label{pr-estimation}
Let $p=\frac{|S|}{2^n}.$ For all $\xi>0,\delta>0$ we can get an estimation  $\hat{p}$ in time $poly(n,\xi^{-1},\log\delta^{-1})$ s.t.
\[\Pr[p-\xi\leq\hat{p}\leq p+\xi]\geq1-\delta\] 
\end{corollary}

\paragraph{Proof of theorem \ref{main-theorem}}

\begin{proof}
1. Follows from propositions \ref{stationary} and \ref{mixing time}. 

2. Follows from theorem \ref{main-theorem}.1.

3a.  Follows from proposition \ref{sizeOfS}, combined with the fact that Size-of-Subtree is TotP-complete \cite{BCPPZ17}.

3b.  Follows from proposition \ref{estimOfaR}. 

3c.  \#IS $\in$ TotP. If NP$\neq$P (respectively NP$\neq$RP) then \#IS does not admit FPTAS (respectively FPRAS) \cite{SinclairNotesMC}. Thus from 3a of theorem 1, the same holds  for the computation of the normalizing factors for the family $(\pi_S)_{_S}.$

4a.  Follows from the TotP-completeness of Size-of-Subtree under parsimonious reductions  \cite{BCPPZ17}. From definition \ref{def-PS},  a positive probability is given to every node of the corresponding input tree $S$, so the size of the support equals the size of the tree.

4d.  By same arguments as for 3c.

4e.  By same arguments as for 3c, for the one direction. For the other direction, if NP=P then \#P admits FPTAS (\cite{AB09}, ch. 17.3.2). If NP=RP then \#P admits FPRAS (theorem \ref{theorem2}).

4c and 4f.  come from corollary \ref{pr-estimation}. 

4b.  follows from proposition \ref{sizeOfS}.
\end{proof}

\section{Discussion}\label{discussion}

We have studied some relationships between counting complexity and \textit{non-uniform} probability distributions. We also studied the \textit{complexity} of some computational tasks related to such distributions. Similar relationships had not been studied before. Some exceptions concern complexity results for individual problems, that do not generalize to a whole class, (see the last paragraph in section \ref{related work}). Our results \textit{ generalize} to all problems in TotP.

We have considered three computational tasks related to any probability distribution: sampling, computing the normalizing factor, and computing the size of the support.  
For the uniform distribution, these three tasks are equivalent \cite{JS89}.
However, for a general distribution, it is not only unknown whether these tasks are solvable in polynomial time; it is even unclear whether these three tasks are equivalent. 

For the family of distributions we studied here (definition \ref{def-PS}), it turns out that the three tasks \textit{ are not all equivalent}, unless NP=RP. First of all counting coincides with computing the size of the support (fact \ref{4a}). Then we showed that sampling is in polynomial time, and that sampling yields an FPRAS for the normalizing factors. So sampling and FPRAS for the normalizing factor are always equivalent, since both are  possible in polynomial time. Also existence of FPRAS for counting, implies existence of FPRAS for the normalizing factor, since the second is always true.  For the opposite direction of the last fact, our results, combined with the fact that FPRAS for TotP is equivalent to NP=RP, imply that: 

\begin{corollary}(FPRAS for the normalizing factors $\Rightarrow$ FPRAS for counting) \textbf{iff }(NP=RP).
\end{corollary}

We also showed that exact counting reduces to exactly computing the normalizing factor (fact \ref{4b}), but not under approximation preserving reductions. The previous arguments imply that such an \textit{ approximation preserving reduction, between the two tasks, exists if and only if NP=RP}.

We now turn to another issue.   Since the Size-of-Subtree problem is TotP complete, under parsimonious reductions, our results generalize to any problem in TotP. However if P $\neq$ NP, we can't derive such a generalization for  \#P. As we mentioned in the introduction, it is not even clear how to construct a Markov chain, with an underlying graph that connects the set of solutions of a \#P- complete problem. Besides we can't decide if a solution exists.  \textit{Our results demonstrate the essence of two properties of TotP: easy decision, and connectedness  of the set of solutions}.

Finally, note that additive error approximation for any problem in \#P can be achieved in a simple way (\cite{Goldreich08}, chapter 6.2.2). Thus the same method works for TotP, too.  The last fact \ref{4last} provides an alternative way for achieving additive error approximation for problems in TotP. This alternative method is restricted to TotP, and does not generalize to \#P. 

A positive side of this restriction is relevant to \textit{derandomization} issues. It is known that derandomizing the general simple method is as difficult as proving circuit lower bounds \cite{Will13}. However we don't know a similar relationship between circuit lower bounds and deterministic additive approximation, restricted to problems in TotP. Since our method does not generalize to \#P, it might be easier to derandomize it. We discuss this in more detail in the "further research" section.

\subsection{Further research}
Several results of this paper point to new research directions, towards the study of important  open problems in complexity theory.

It is interesting to investigate the following: (a) probabilistically exactly  computing  the normalizing factor, (b) an approximation preserving reduction from the problem of computing the size of the support to the problem of computing the normalizing factor, (c) computing the size of the support for this family of distributions with FPRAS in some completely different way. A solution to any of them implies NP=RP. A negative proof for (b) or (c) implies NP $\neq$ RP.

 Note that the algorithms and proof arguments presented here, do not take into account the fact that the tree is given in succinct representation. It might be easier to show unconditional negative results for the above open questions, if we suppose $S$ is any possible binary tree. However, in order to derive a proof of NP $\neq$ RP, the arguments should apply to the family of distributions corresponding to trees with succinct representation.

Another  open problem is how to derandomize the additive error approximation algorithms for  the size of the support, in subexponential time. This would yield a deterministic solution of the CAPE problem (see def. \ref{additive definition}), within additive approximation, for families of circuits, for which  counting the number of accepting inputs is in TotP, (we will call it TotP-CAPE.)

Note that the best (exact, and additive error) deterministic algorithm known for CAPE, on an arbitrary circuit, is by exhaustive search. Derandomizing it faster than the exhaustive search algorithm, i.e.  even in time $2^{\gamma n}poly(n)$  for some $\gamma<1$,  would yield NEXP $\nsubseteq$ P/poly \cite{Will13}. The latter is a long standing  conjecture. 

We don't know similar relationships between circuit lower bounds, and TotP-CAPE. This is another open problem. Nevertheless, solving deterministically TotP-CAPE in subexponential time, would be interesting on its own.   

A final open problem, is whether we can achieve derandomization of the same task in polynomial time.  Such a result would also imply a solution to the CAPE problem in deterministic polynomial time for depth-two $AC^0$ circuits (i.e. DNF's and CNF's). The best deterministic algorithm known until now is of time $n^{O(\log \log n)}$ \cite{GMR13}. (For more on $AC^0$-CAPE, see the survey \cite{Williams14} p.13, and \cite{LV96} for an older result.)

\section{Related work}\label{related work}

TotP is defined in \cite{KPSZ01}, some of its properties and relationships to other classes are studied in \cite{PZ06,
BGPT17}, and completeness is studied in \cite{BCPPZ17}.

Regardint the Size-of-Subtree: In \cite{K74} Knuth provides a probabilistic algorithm practically usefull, but with an exponential worst case error. Modifications and extensions of Knuth's algorithm have been presented and experimentaly tested in \cite{Purdom78, Chen92, Kilby06}, without signifcant improvements for worst case instances. There are also many heuristics and experimental results on the problem ristricted to special backtracking algorithms, or special instances, see e.g. \cite{Belov17} for more references. Surprisingly there exist FRAS's for random models of the problem \cite{Furer04, Vaisman17}. In \cite{Ambainis17} quantum algorithms for  the problem are studied. In \cite{Stockmeyer85a} Stockmeyer provided unconditional lower bounds for the problem under a model of computation which is different from the Turing machine, namely a variant of the (non-uniform) decision tree model.

The relationship between approximate counting and uniform sampling has been studied in \cite{JS89}.

There are numerous papers regarding algorithmic and hardness results for individual problems in \#P and TotP, e.g. \cite{Valiant79,JS96perm,KLM89,DFJ02,Wei06}. However, apart from the backtracking-tree problem,  other TotP-complete problems have not been studied algorithmically yet.

Some non-uniform probability measures have already been studied in other  areas of computer science, where problems concern the computation of  a weighted sum over the set of solutions to a combinatorial problem. E.g. computing the partition function of the hard-core model and the Potts-model from statistical physics. There are Markov chains associated to these problems, that converge to non-uniform distributions in general. E.g. Glauber dynamics, Gibbs sampling \cite{Gibbs}, Metropolis-Hastings algorithm \cite{Metropolis,Hastings} etc. However, for the special cases where weights are in $\{0,1\}$, the problems correspond to  conventional counting problems in \#P, and the associated stationary distributions are uniform, again. The literature on this areas is enormous, e.g. \cite{BST10,BKZZ13,BW02,So0,WM0,W82}.

\section{Conclusions} 
We presented some non-uniform probability measures that can be related to counting in TotP. We showed that both computing the size of their supports, and computing their normalizing factors are equivalent to counting. 

For these probability measures we proved that, the tasks of sampling, approximately computing the normalizing factor with FPRAS, and  approximating the size of the support with an additive error, can be performed in randomized polynomial time.

We also showed that an exact computation of the normalizing factor, and a multiplicative error approximation of the size of the support, are hard problems if NP $\neq$ RP.

Such relationships between counting complexity and \textit{non-uniform} probability measures had not been studied before. 

Our results apply to the whole complexity class TotP. Similar results are not possible for \#P, if P $\neq$ NP. Our algorithmic results \textit{demonstrate the importance of two main properties of TotP; easy decision and connectedness of the set of solutions}.

Our results also suggest new research directions towards the study of other open problems in complexity theory.

\section*{Appendix}
We don't know if the following is folklore, but since we haven't seen it stated explicitly anywhere, we give a proof sketch, for the sake of completeness of this paper.

\begin{theorem}\label{theorem2}
If NP=RP then all problems in \#P admit an FPRAS.
\end{theorem}
\begin{proof} (sketch)
In \cite{Stockmeyer85a} Stockmeyer has proven that an FPRAS, with access to a $\Sigma_2^p$ oracle, exists for any problem in \#P. If NP=RP then $\Sigma_2^p= RP^{RP}\subseteq BPP$. Finally it is easy to see that an FPRAS with access to a BPP oracle, can be replaced by another  FPRAS, that simulates the oracle calls itself.
\end{proof}

\section*{Acknowledgements}
I want to thank Stathis Zachos, Dimitris Fotakis, Aris Pagourtzis, Manolis Zampetakis and Gerasimos Palaiopanos for their useful comments and writing assistance.

\end{document}